\documentclass{llncs}

\usepackage[intlimits]{amsmath}
\usepackage{amsfonts,amssymb}

\usepackage{amsthm}
\usepackage{enumerate}

\makeatletter
\newtheorem*{rep@theorem}{\rep@title}
\newcommand{\newreptheorem}[2]{%
\newenvironment{rep#1}[1]{%
 \def\rep@title{#2 \ref{##1}}%
 \begin{rep@theorem}}%
 {\end{rep@theorem}}}
\makeatother
\newreptheorem{theorem}{Theorem}
\newreptheorem{lemma}{Lemma}
\newreptheorem{corollary}{Corollary}

\begin{document}

\title{Size of Sets with Small Sensitivity: a Generalization of Simon's Lemma\thanks{
The research leading to these results has received funding from the European Union Seventh Framework Programme (FP7/2007-2013) under projects QALGO (Grant Agreement No. 600700) and RAQUEL (Grant Agreement No. 323970)
and ERC Advanced Grant MQC. 
Part of this work was done while Andris Ambainis was
visiting Institute for Advanced Study, Princeton, supported by National Science 
Foundation under agreement No. DMS-1128155. Any opinions, findings and conclusions or recommendations expressed in this material are those of the author(s) and 
do not necessarily reflect the views of the National Science Foundation.}}
\author{Andris Ambainis \and Jevg\={e}nijs Vihrovs}
\institute{Faculty of Computing, University of Latvia, Rai\c{n}a bulv. 19, R\=\i ga, LV-1586, Latvia}

\date{}

\maketitle

\begin{abstract}
We study the structure of sets $S\subseteq\{0, 1\}^n$ with small sensitivity. The well-known Simon's lemma says that any $S\subseteq\{0, 1\}^n$
of sensitivity $s$ must be of size at least $2^{n-s}$. This result has been useful for proving lower bounds on the sensitivity of Boolean functions,
with applications to the theory of parallel computing and the ``sensitivity vs. block sensitivity'' conjecture.

In this paper we take a deeper look at the size of such sets and their structure. We show an unexpected ``gap theorem'': if $S\subseteq\{0, 1\}^n$ has sensitivity $s$, then we either have $|S|=2^{n-s}$ or $|S|\geq \frac{3}{2} 2^{n-s}$. 

This provides new insights into the structure of low sensitivity subsets of the Boolean hypercube $\{0, 1\}^n$.
\end{abstract}

\section{Introduction}

The complexity of computing Boolean functions (for example, in the decision tree model of computation) is related to a number of combinatorial quantities, such as
the sensitivity and block sensitivity of the function, its certificate complexity and the degree of polynomials that represent the function exactly or approximately \cite{Buhrman_deWolf_2002}.
Study of these quantities has resulted in both interesting results and longstanding open problems.

For example, it has been shown that decision tree complexity in either a deterministic, a probabilistic or a quantum model of computation is polynomially related to a number
of these quantities: certificate complexity, block sensitivity and the minimum degree of polynomials that represent or approximate $f$ \cite{Nisan_Szegedy_1994,Beals+_2001}. This result, in turn, implies that deterministic, 
probabilistic and quantum decision tree complexities are polynomially related --- which is very interesting because a similar result is not known in the Turing
machine world; and, for deterministic vs. quantum complexity, is most likely false because of Shor's factoring algorithm.

The question about the relation between the sensitivity of a function and the other quantities is, however, a longstanding open problem,
known as the ``sensitivity vs. block sensitivity'' question. Since the other quantities are all polynomially related, showing a polynomial relation between sensitivity
and any one of them would imply a polynomial relation between sensitivity and all of them.
This question, since first being posed by Nisan in 1991 \cite{Nisan_1991}, has attracted much attention but there has been quite little progress
and the gap between the best upper and lower bounds remains huge. The examples that achieve the asymptotically biggest separation between the two quantities give $bs(f)=\Omega(s^2(f))$ \cite{Ambainis_Sun_2011,Rubinstein_1995,Virza_2011}, while the best upper bound on $bs(f)$ in terms of $s(f)$ is exponential: $bs(f) \leq s(f) 2^{s(f)-1}$ \cite{A+,Kenyon_Kutin_2004}. Here $bs(f)$ and $s(f)$ denote the block sensitivity and the sensitivity of $f$, respectively.

In this paper we study the following question: assume that a subset $S$ of the Boolean hypercube $\{0, 1\}^n$ has low sensitivity: that is, for every $x\in S$ there are at most
$s$ indices $i\in\{1, \ldots, n\}$ such that changing $x_i$ to the opposite value results in $y\notin S$. What can we say about this set?

Most of the upper bounds on $bs(f)$ in terms of $s(f)$ are based on Simon's lemma \cite{Simon_1983}.
We say that a subset $S$ of the Boolean hypercube $\{0, 1\}^n$ has sensitivity $s$ if, for every $x\in S$, there are at most
$s$ indices $i\in\{1, \ldots, n\}$ such that changing $x_i$ to the opposite value results in $y\notin S$.
Simon's lemma \cite{Simon_1983} says that any $S\subset \{0, 1\}^n$ with sensitivity $s$ must contain at least $2^{n-s}$ 
input vectors $x\in S$. 

Simon \cite{Simon_1983} then used this result to show that 
$s(f)\geq\frac{1}{2} \log_2 n - \frac{1}{2} \log_2 \log_2 n + \frac{1}{2}$ 
for any Boolean function that depends on $n$ variables. Since $bs(f)\leq n$, 
this implies $bs(f)\leq s(f) 4^{s(f)}$. This was the first upper bound on $bs(f)$ in terms of $s(f)$. 
A more recent upper bound of $bs(f) \leq s(f) 2^{s(f)-1}$ by Ambainis et al. \cite{A+} is also based on Simon's lemma. 
If it was possible to improve Simon's lemma, this would result in better bounds on $bs(f)$.

However, Simon's lemma is known to be exactly optimal. 
Let $S$ be a subcube of the hypercube $\{0, 1\}^n$ obtained by fixing $s$ of variables $x_i$. That is, $S$ is the set of all $x=(x_1, \ldots, x_n)$ that
satisfy $x_{i_1}=a_1$, $\ldots$, $x_{i_s}=a_s$ for some choice of distinct $i_1, \ldots, i_k\in\{1, \ldots, n\}$ and 
$a_1, \ldots, a_s\in\{0, 1\}$. Then every $x\in S$ is sensitive to changing $s$ bits $x_{i_1}, \ldots, x_{i_k}$ and 
$|S|=2^{n-s}$.

In this paper, we discover a direction in which Simon's lemma can be improved!
Namely, we show that any $S$ with sensitivity $s$ that is not a subcube must be substantially larger. To do that, we
study the structure of sets $S$ with sensitivity $s$ by classifying them into two types:
\begin{enumerate}
\item
sets $S$ that are contained in a subcube $S'\subset \{0, 1\}^n$ obtained by fixing one or more of values $x_i$;
\item
sets $S$ that are not contained in any such subcube.
\end{enumerate}

There is one-to-one correspondence between the sets of the first type and low-sensitivity subsets of $\{0, 1\}^{n-k}$ for $k\in\{1, \ldots, s\}$.\footnote{If a set $S$
of sensitivity $s$ is contained in a subcube $S'$ obtained by fixing $x_{i_1}, \ldots, x_{i_{k}}$, removing the variables that have been fixed gives us
a set $S''\subseteq \{0, 1\}^{n-k}$ of sensitivity $s-k$.} In contrast, the sets of the second type do not reduce to low-sensitivity subsets of $\{0, 1\}^{n-k}$ for $k>0$.
Therefore, we call them {\em irreducible}.

Our main technical result (Theorem \ref{thm:main}) is that any irreducible $S\subseteq \{0, 1\}^n$ must be of size $|S|\geq  2^{n-s+1} - 2^{n-2s}$, almost twice as large as a subcube obtained
by fixing $s$ variables, and this bound is tight.

As a consequence, we obtain a surprising result: if $S\subseteq \{0, 1\}^n$ has sensitivity $s$, then either $|S|=2^{n-s}$ or $|S|\geq \frac{3}{2} 2^{n-s}$. 
That is, such a set $S$ cannot have a size between $2^{n-s}$ and $\frac{3}{2} 2^{n-s}$  (Theorem \ref{thm:stronger}).

In a following work \cite{Ambainis_Prusis_Vihrovs}, we have applied this theorem to obtain a new upper bound on block sensitivity in terms of sensitivity:
\begin{equation}
bs(f) \leq \max\left(2^{s(f)-1}\left(s(f)-\frac 1 3\right),s(f) \right).
\end{equation}

{\bf Related work.}
A gap theorem of a similar type is known for the spectral norm of Boolean functions \cite{green2008boolean}: the spectral norm of a Boolean function is either equal to 1 or is at least $\frac 3 2$. Both results have the constant $\frac 3 2$ appearing in them and there is some resemblance between the constructions of optimal sets/functions but the proof methods are
quite different and it is not clear to us if there is a more direct connection between
the results.

\section{Preliminaries}

In this section we give the basic definitions used in the paper. Let $f : \{0,1\}^n \rightarrow \{0,1\}$ be a Boolean function of $n$ variables, where the 
$i$-th variable is denoted by $x_i$. We use $x=(x_1, \ldots, x_n)$ to denote
a tuple consisting of all input variables $x_i$.

\begin{definition}
The \emph{sensitivity complexity} $s(f,x)$ of $f$ on an input $x$ is defined as $| \{i \mid f(x) \neq f(x^{(i)})\} |$, where $x^{(i)}$ is an input obtained from $x$ by flipping the value of the $i$-th variable. The \emph{sensitivity} $s(f)$ of $f$  is defined as 
\begin{equation}
s(f) = \max \{s(f,x) \mid x \in \{0,1\}^n\}.
\end{equation}
The \emph{$c$-sensitivity} $s_c(f)$ of $f$ is defined as 
\begin{equation}
s_c(f) = \max \{s(f,x) \mid x \in \{0,1\}^n, f(x)=c\}.
\end{equation}
\end{definition}

In this paper we will look at $\{0, 1\}^n$ as a set of vertices for a graph $Q_n$
(called the {\em $n$-dimensional Boolean cube} or {\em hypercube})
in which we have an edge $(x, y)$ whenever $x=(x_1, \ldots, x_n)$ and 
$y=(y_1, \ldots, y_n)$ differ in exactly one position.
We look at subsets $S \subseteq \{0,1\}$ as subgraphs (induced by the subset of vertices $S$) in this graph. 

\begin{definition}
We define an $m$-dimensional \emph{subcube} or $m$-\emph{subcube} of $Q_n$ to be a cube induced by the set of all vertices that have the same bit values on $n-m$ positions $x_{i_1}, \ldots, x_{i_{n-m}}$ where $i_j$ are all different.
\end{definition}

We denote a subcube that can be obtained by fixing some continuous sequence $b$ of starting bits by $Q_b$. For example, $Q_0$ and $Q_1$ can be obtained by fixing the first bit and $Q_{01}$ can be obtained by fixing the first two bits to 01. We use a wildcard * symbol to indicate that the bit in the corresponding position is not fixed. For example, by $Q_{*10}$ we denote a cube obtained by fixing the second and the third bit to 10.

\begin{definition}
Two $m$-dimensional subcubes of $Q_n$ are \emph{adjacent} if the fixed $n-m$ positions of both subcubes are the same and their bit values differ in exactly one position.
\end{definition}

Each Boolean function $f$ can be uniquely represented as a set of vertices $V(f) = \{ x \mid f(x) = 1\}$, thus each function of $n$ variables represents a single subgraph $G(f)$ of $Q_n$ induced by $V(f)$. Note that for an input $x \in V(f)$, the sensitivity $s(f,x)$ is equal to the number of vertices not in $V(f)$ and connected to $x$ with an edge in $Q_n$. Thus the sensitivity of $V(f)$ is equal to $s_1(f)$.

For a Boolean function $f$, the minimum degree $\delta(G(f))$ corresponds to $n-s_1(f)$, and the minimum degree of a graph induced by $\{0,1\}^n \setminus V$ corresponds to $n - s_0(f)$.

In the rest of this paper we phrase our results in terms of subgraphs of $Q_n$.

\begin{definition}
Let $X$ and $Y$ be subgraphs of $Q_n$. By $X \cap Y$ we denote the \emph{intersection} graph of $X$ and $Y$ that is the graph $(V(X) \cap V(Y), E(X) \cap E(Y))$. By $X \setminus Y$ denote the \emph{complement} of $Y$ in $X$ that is the graph induced by the vertex set $V(X) \setminus V(Y)$ in $X$.
\end{definition}

We also denote the degree of a vertex $v$ in a graph $G$ by $\deg(v,G)$.

\bigskip

The main focus of the paper is on the \emph{irreducible} class of subgraphs:

\begin{definition}
We call a subgraph $G \subset Q_n$ \emph{reducible} if it is a subgraph of some graph $S \subset Q_n$ where $V(S)$ can be obtained by fixing one or more of values $x_i$. Conversely, other subgraphs we call \emph{irreducible}.
\end{definition}

Another way to define the irreducible graphs is to say that each such graph contains at least one vertex in each of the $(n-1)$-subcubes of $Q_n$.

\section{Simon's Lemma}

In this section we present a theorem proved by Simon \cite{Simon_1983}. 


\begin{theorem}[Simon] \label{thm:simon}
Let $G = (V, E)$ be a non-empty subgraph of  $Q_n$ ($n \geq 0$) of minimum cardinality among the subgraphs with $\delta(G) = d$ ($d \geq 0$). Then $G$ is a $d$-dimensional subcube of $Q_n$ and $|V| = 2^d$.
\end{theorem}


This theorem implies:

\begin{corollary}
Let $f(x)$ be a Boolean function on $n$ variables. If $f(x)$ is not always 0, then
\begin{equation}
| \{ x \mid f(x) = 1 \} | \geq 2^{n-s_1(f)},
\end{equation}
and the minimum is obtained iff some $s_1(f)$ positions hold the same bit values for all $x : f(x) = 1$.
\end{corollary}

\begin{proof}
Let $G$ be a subgraph of $Q_n$ induced by the set of vertices $V = \{ x \mid f(x) = 1 \}$. The minimum degree of $G$ is $\delta(G) = n - s_1(f)$. Then by Theorem \ref{thm:simon} $|V| \geq 2^{n-s_1(f)}$. The minimum is obtained iff $G$ is an $(n-s_1(f))$-subcube of $Q_n$. This means that it is defined by some bits fixed in $s_1(f)$ positions.
\end{proof}

\section{Smallest Irreducible Subgraphs}

In this section we prove the main theorem.

\begin{theorem} \label{thm:main}
Let $G = (V, E)$ be a non-empty irreducible subgraph of $Q_n$ ($n \geq 1$) with the minimum degree $d \geq 0$. Let the smallest possible cardinality of $V$ be $S(n,d)$. Then
\begin{equation}
S(n,d) = \left\lceil 2^{d+1} - 2^{2d-n} \right\rceil.
\end{equation}
\end{theorem}

The proof of Theorem \ref{thm:main} is by induction on $n$ and involves case analysis going as deep as considering $(n-3)$-dimensional subcubes of $Q_n$.

In the language of Boolean functions, this theorem corresponds to:

\begin{corollary}\label{thm:result}
Let $f(x)$ be a Boolean function on $n$ variables. If $\forall i \in [n] \, \forall b \in \{0, 1\} \, \exists x \, (x_i = b, f(x) = 1)$, then
\begin{equation}
| \{ x \mid f(x) = 1 \} | \geq 2^{n-s_1(f)+1} - 2^{n-2s_1(f)}.
\end{equation}
\end{corollary}

Theorem \ref{thm:main} together with Lemma \ref{thm:minsize} imply the following generalization of Simon's lemma:

\begin{theorem}\label{thm:stronger}
Let $G = (V, E)$ be a non-empty subgraph of $Q_n$ ($n \geq 0$) with $\delta(G) = d$. Then either $|V| = 2^d$ or $|V| \geq \frac 3 2 \cdot 2^d$, with $V=|2^d|$ achieved if and only if $G$ is a $d$-subcube.
\end{theorem}

Equivalently, if $G$ has sensitivity $s$, then either $|V| = 2^{n-s}$ or $|V| \geq \frac 3 2 2^{n-s}$. Thus there is a gap between the possible values for $|V|$ --- which we find quite surprising.

In the next two subsections we prove Theorem \ref{thm:main} and in the last two subsections we show how it implies Corollary \ref{thm:result} and  Theorem \ref{thm:stronger}.

\subsection{Instances Achieving the Minimum} \label{sec:sufficiency}

In this section we prove that the given number of vertices is sufficient. We distinguish three cases:
\begin{enumerate}
	\item $n = 1$. The only valid graph satisfying the properties is $G = Q_n$ with $d = 1$. Then $|V| = 2$.

	\item $n > 1$, $2d < n$. Since $2^{2d-n} < 1$, $|V|$ should be $2^{d+1}$. We take 
\begin{equation}
S_j = \{ x \mid \forall i \in [n-d] \, (x_i = j) \}
\end{equation}
for $j \in \{0, 1\}$ and $V = S_0 \cup S_1$. Let $G$ be the graph induced by $V$ in $Q_n$. Then $G$ consists of two $d$-subcubes of $Q_n$ with no common vertices. Since $n-d > 1$, no edge connects any two vertices between these subcubes, thus $\delta(G) = d$. For the irreducibility, suppose that some $(n-1)$-subcube $H$ is defined by fixing $x_i = j$. If $i \leq n - d$, then $H \cap S_j \neq \varnothing$. If $i > n - d$, then $H \cap S_j \neq \varnothing$ for any $j$. Then $|V| = 2\cdot 2^d = 2^{d+1}$.

	\item $n > 1$, $2d \geq n$. Then $|V|$ should be $2^{d+1} - 2^{2d-n}$. We take 
\begin{align}
S_l &= \{ x \mid \forall i \in [n-d] \, (x_i = 1) \},\\
S_r &= \{ x \mid \forall i \in [n-d+1; 2(n-d)] \, (x_i = 1) \}
\end{align}
and $V = S_l \cup S_r$. Let $G$ be the graph induced by $V$ in $Q_n$. Graphs induced by $S_l$ and $S_r$ are $d$-dimensional subcubes of $Q_n$. Since they are not adjacent, $\delta(G) = d$. For the irreducibility, observe that any bit position $i$ is not fixed for at least one of $S_l$ or $S_r$. Then the $(n-1)$-subcube $H$ obtained by fixing $x_i$ holds at least one of the vertices of $G$. Since $S_l \cap S_r = \{ x \mid \forall i \in [2(n-d)] (x_i = 1) \}$, it follows that 
\begin{equation}
|V| = 2\cdot 2^d - 2^{n-2(n-d)} = 2^{d+1} - 2^{2d-n}.
\end{equation}
\end{enumerate}

\subsection{Optimality}

In this section we prove that there are no such graphs with a number of vertices less than $\left\lceil 2^{d+1} - 2^{2d-n} \right\rceil$.

The proof is by induction on $n$. As the base case we take $n \leq 2$. From the fact that each $(n-1)$-subcube contains at least one vertex of $G$ it follows that $|V| \geq 2$. This proves the cases $n = 1$, $d = 1$ and $n=2$, $d = 0$ (and the case $n=1$, $d=0$ is not possible). Suppose $n=2$, $d=1$: if there were 2 vertices in $G$, then either some of the 1-subcubes would contain no vertex of $G$ or there would be a vertex of $G$ with degree 0 (which contradicts $d=1$). Thus, in this case $|V| \geq 3 = 2^{1+1} - 2^{2-2}$. Suppose $n=2$, $d=2$. Then $G = Q_n$ and $|V| = 4 = 2^{2+1} - 2^{4-2}$.

Inductive step. First suppose that each $(n-2)$-subcube of $Q_n$ contains at least one vertex of $G$, then $G\cap Q_0$ and $G \cap Q_1$ are irreducible. The minimum degrees of $G\cap Q_0$ and $G \cap Q_1$ are at least $d-1$, since each vertex of $G\cap Q_0$ can have at most one neighbour in $Q_1$ (and conversely). By applying the inductive assumption to the cubes $Q_0$ and $Q_1$, we obtain that 
\begin{align}
|V| &\geq 2 \cdot \left\lceil 2^{(d-1)+1} - 2^{2(d-1)-(n-1)} \right\rceil =\\
&= 2 \cdot \left\lceil 2^{d} - 2^{2d-n-1} \right\rceil \geq \\
&\geq \left\lceil 2^{d+1} - 2^{2d-n} \right\rceil.
\end{align}

Now suppose that there is some $(n-2)$-subcube without vertices of $G$. WLOG assume it is $Q_{00}$, i.e. $G \cap Q_{00} = \varnothing$. We prove two lemmas.

\begin{lemma} \label{thm:minsize}
Let $G = (V, E)$ be a non-empty subgraph of $Q_n$ ($n \geq 0$) with $\delta(G) = d$ ($d \geq 0$). Then either $|V| = 2^d$ or $|V| \geq \min_{i=d+1}^n S(i, d)$.
\end{lemma}

\begin{proof}
The proof is by induction on $n$. Base case: $n = 0$. Then $G = Q_n$, $d = 0$ and $|V| = 1 = 2^{0-0}$. In the inductive step we prove the statement for $n > 0$. If $n = d$, then $G = Q_n$, and $|V| = 2^n = 2^d$. Otherwise $n > d$. If each $(n-1)$-subcube of $Q_n$ contains vertices of $G$, then $|V| \geq S(n,d)$ by the definition of $S$. Otherwise there is an $(n-1)$-subcube of $Q_n$ that does not contain any vertex of $G$. Then by induction the other $(n-1)$-subcube contains either $2^d$ or at least $\min_{i=d+1}^{n-1} S(i, d)$ vertices of $G$. Combining the two cases together gives us the result.
\end{proof}

\begin{lemma} \label{thm:extended}
Let $G = (V, E)$ be a subgraph of $Q_n$ ($n \geq 1$). Let $G' = G \cap Q_0$. If $G'$ is not empty and $\min_{v\in G'}\deg(v, G) \geq d$, then $|V| \geq 2^d$.
\end{lemma}

Note that this lemma is also a stronger version of Simon's result. Here we require the lower bound for the minimum degree only for vertices of $G$ in one of the $(n-1)$-subcubes of $Q_n$.

\begin{proof}
The proof is by induction on $n$.
\begin{enumerate}[(a)]
	\item Base case, $n = 1$. Since $G'$ is non-empty, $G' = Q_0$. If $d = 0$, $|V| \geq 1 = 2^0$. If $d = 1$, then $G = Q_n$ and $|V| = 2 = 2^1$.
	\item In the inductive step we prove the statement for $n > 1$. If $Q_{0j} \cap G'$ is empty for some $j \in \{0,1\}$, then $G' \subseteq Q_{0(1-j)}$. Thus by the induction hypothesis $|V(Q_{*(1-j)})| \geq 2^d$. Otherwise both $Q_{00}$ and $Q_{01}$ contain some vertices of $G$. Since each vertex of $Q_{0j} \cap G$ has at most one neighbour in $Q_{0(1-j)} \cap G$, it follows that $\min_{v\in Q_{0j} \cap G}\deg(v, Q_{*j}) \geq d-1$ for any $j \in \{0, 1\}$. By applying the induction hypothesis for $Q_{*j} \cap G$ in the cube $Q_{*j}$ for each $j$, we obtain that $|V| \geq 2 \cdot 2^{d-1} = 2^d$.
\end{enumerate}
\end{proof}

We now have that $\delta(G \cap Q_{01}) \geq d-1$ and $\delta(G \cap Q_{10}) \geq d-1$ becase $Q_{11}$ may contain vertices of $G$ but on the other hand we are assuming $G \cap Q_{00} = \varnothing$. Now we distinguish two cases:

\renewcommand{\labelenumi}{\textbf{\arabic{enumi}.}} 
\renewcommand{\labelenumii}{\textbf{\arabic{enumi}.\arabic{enumii}.}}
\renewcommand{\labelenumiii}{\textbf{\arabic{enumi}.\arabic{enumii}.\arabic{enumiii}.}}

\begin{enumerate}

\item $|V(G \cap Q_{01})| \neq 2^{d-1}$ and $|V(G \cap Q_{10})| \neq 2^{d-1}$.

Cube $Q_{01}$ has $n-2$ dimensions and $\delta(Q_{01} \cap G) \geq d-1$. By Lemma \ref{thm:minsize}
\begin{equation}
|V(Q_{01} \cap G)| \geq \min_{i=(d-1)+1}^{n-2} S(i, d-1) = \min_{i=d}^{n-2} S(i, d-1).
\end{equation}
It follows by induction that
\begin{equation}
|V(Q_{01} \cap G)| \geq  \min_{i=d}^{n-2} \left\lceil 2^{(d-1)+1} - 2^{2(d-1)-i} \right\rceil.
\end{equation}
The minimum is achieved when $i$ is the smallest, $i = d$. Thus $|V(Q_{01} \cap G)| \geq \left\lceil 2^d - 2^{d-2} \right\rceil$. Similarly we prove that $|V(Q_{10} \cap G)| \geq \left\lceil 2^d - 2^{d-2} \right\rceil$.

\bigskip

It remains to estimate the number of vertices of $G$ in $Q_{11}$. We deal with two cases:
\begin{enumerate}
\item Some $(n-3)$-subcube of $Q_n$ in $Q_{11}$ does not contain vertices of $G$. WLOG we assume it is $Q_{110}$, i.e., $G \cap Q_{110} = \varnothing$. We again distinguish two cases:
\begin{enumerate}
\item One of the subcubes $Q_{010}$ and $Q_{100}$ does not contain vertices of $G$. WLOG assume it is $Q_{010}$, i.e., $G \cap Q_{010} = \varnothing$. Then for the subcube $Q_{011}$ it holds that $\min_{v\in G \cap Q_{011}}\deg(v, G \cap Q_{*11}) \geq d$, since $G \cap Q_{001} = \varnothing$ (because $Q_{001} \subset Q_{00}$), $G \cap Q_{010} = \varnothing$ and $Q_{111}$ may contain vertices of $G$. Applying Lemma \ref{thm:extended} to $G \cap Q_{011}$ in $Q_{*11}$, we get $|V(G \cap Q_{*11})| \geq 2^d$. Similarly we prove that $|V(G \cap Q_{10*})| \geq 2^d$. That gives us
\begin{equation}
|V| \geq 2 \cdot 2^d = 2^{d+1} \geq \left\lceil 2^{d+1} - 2^{2d-n} \right\rceil
\end{equation}
and the case is done.

\item Both of the subcubes $Q_{010}$ and $Q_{100}$ contain vertices of $G$. Then for the subcube $Q_{010}$ it holds that $\min_{v\in G \cap Q_{010}}\deg(v, G \cap Q_{01*}) \geq d$, since $G \cap Q_{000} = \varnothing$, $G \cap Q_{110} = \varnothing$, and $Q_{011}$ may contain vertices of $G$. Applying Lemma \ref{thm:extended} to $G \cap Q_{010}$ in $Q_{01*}$, we get $|V(G \cap Q_{01*})| \geq 2^d$. Similarly we prove that $|V(G \cap Q_{10*})| \geq 2^d$. That gives us
\begin{equation}
|V| \geq 2 \cdot 2^d = 2^{d+1} \geq \left\lceil 2^{d+1} - 2^{2d-n} \right\rceil
\end{equation}
and this case also is done.
\end{enumerate}
\item Each $(n-3)$-subcube of $Q_n$ in $Q_{11}$ contains vertices of $G$. Since $Q_{11}$ is adjacent to $Q_{01}$ and $Q_{10}$, $\delta(G \cap Q_{11}) \geq d-2$. From the inductive assumption it follows that 
\begin{equation}
|V(G \cap Q_{11})| \geq 2^{(d-2)+1} - 2^{2(d-2)-(n-2)} = 2^{d-1} - 2^{2d-n-2}.
\end{equation} Thus
\begin{align}
|V| &= |V(G \cap Q_{01})| + |V(G \cap Q_{10})| + |V(G \cap Q_{11})| \geq \\
& \geq 2 \cdot \left\lceil2^d - 2^{d-2}\right\rceil + \left\lceil2^{d-1} - 2^{2d-n-2}\right\rceil \geq \\
&\geq \left\lceil2 \cdot \left(2^d - 2^{d-2}\right) + 2^{d-1} - 2^{2d-n-2}\right\rceil = \\
&= \left\lceil 2^{d+1} - 2^{d-1} + 2^{d-1} - 2^{2d-n-2}\right\rceil = \\
&= \left\lceil 2^{d+1} - 2^{2d-n-2}\right\rceil \geq \\
&\geq \left\lceil 2^{d+1} - 2^{2d-n}\right\rceil.
\end{align}
Hence this case is complete.
\end{enumerate}

\item $|V(G \cap Q_{01})| = 2^{d-1}$ or $|V(G \cap Q_{10})| = 2^{d-1}$. WLOG assume that this holds for $Q_{01}$.

By Theorem \ref{thm:simon} it follows that $G \cap Q_{01}$ is a $(d-1)$-dimensional subcube of $Q_n$, denote it by $D_0$. On the other hand, we are assuming $G \cap Q_{00} = \varnothing$. Thus WLOG we can assume that $D_0$ is induced on the set of vertices
\begin{equation}
\{ x \mid x_1 = 0, \forall i \in [2;n-d+1] \, (x_i = 1) \} = V(G \cap Q_0).
\end{equation}
Observe that $\deg(v,G \cap Q_{01}) = d-1$ for all $v \in G \cap Q_{01}$. Since $\delta(G) = d$, each $x \in V(G \cap Q_{01})$ has $x^{(1)}$ as a neighbour in $G$. Then $\{ x^{(1)} \mid x \in V(G \cap Q_{01}) \} \subseteq V(G \cap Q_{11})$, and $G \cap Q_{11}$ contains a $(d-1)$-subcube of $Q_n$ adjacent to $D_0$. We denote it by $D_1$, with
\begin{equation}
\{ x \mid x_1 = 1, \forall i \in [2;n-d+1] \, (x_i = 1) \} \subseteq V(G \cap Q_1).
\end{equation}
Then $D = D_0 \cup D_1$ is a $d$-dimensional subcube.

It remains to estimate the number of vertices of $G$ in $Q_1$ that do not belong to $D_1$, denote it by $R = |V((G \cap Q_1) \setminus D_1)|$. We will prove the following claim:
\begin{claim}
By $k$ denote the co-dimension of $D_1$ in $Q_1$, which is $(n-1)-(d-1) = n-d$. Then $R \geq 2^d - 2^{d - k}$.
\end{claim}

\begin{proof}
We will denote the subcube of $Q_1$ obtained by restricting some $t$ bits $x_{i_1}=b_1$, $\ldots$, $x_{i_t}=b_t$ by $Q_1(x_{i_1}=b_1,\ldots,x_{i_t}=b_t)$. Further note that $D_1 \subseteq Q_1(x_i=1,x_j=1)$ for $i, j \in [2;k+1]$. 

Since $G \cap Q_0 = D_0$, any vertex of $(G \cap Q_1) \setminus D_1$ can have a neighbour in $G$ only in $Q_1$. Thus we have that
\begin{equation}
\min_{v \in (G \cap Q_1) \setminus D_1} \deg(v, G \cap Q_1) \geq d.
\end{equation}

Pick any $i \in [2;k+1]$. Examine the $(n-2)$-subcube $Q_1(x_i = 0)$. It does not overlap with $D$. But $G$ is irreducible, so $G \cap Q_1(x_i = 0) \neq \varnothing$.

Assume $k = 1$. Then $D_1 = Q_{11}$ and $\delta(G \cap Q_{10}) = d-1$. By Theorem \ref{thm:simon}, it follows that
\begin{equation}
R = |V(G \cap Q_{10})| \geq 2^{d-1} = 2^d - 2^{d-1}.
\end{equation}

Otherwise $k \geq 2$. We will prove it can be assumed that for any $i, j \in [2;k+1]$, $i \neq j$ and $b \in \{0,1\}$ we have $G \cap Q_1(x_i=0,x_j=b) \neq \varnothing$.

\begin{itemize}
\item Let $G \cap Q_1(x_i = 0, x_j = 0) = \varnothing$. Then $\delta(G \cap Q_1(x_i = 0, x_j = 1)) \geq d-1$ and $\delta(G \cap Q_1(x_i = 1, x_j = 0)) \geq d-1$. By Theorem \ref{thm:simon}, we have $|V(G \cap Q_1(x_i = 0, x_j = 1))| \geq 2^{d-1}$ and $|V(G \cap Q_1(x_i = 1, x_j = 0))| \geq 2^{d-1}$. Thus in this case
\begin{align}
R \geq 2\cdot2^{d-1} = 2^d > 2^d - 2^{d-k}.
\end{align}

\item Let $G \cap Q_1(x_i = 0, x_j = 1) = \varnothing$. Then $\min_{v \in G \cap Q_1(x_i = 0, x_j = 0)} \deg(v,G\cap Q_1(x_j=0)) \geq d$ and since $G \cap Q_1(x_j=0) \neq \varnothing$, by Lemma \ref{thm:extended} we have
\begin{equation}
R > |V(G \cap Q_1(x_j = 0))| \geq 2^d > 2^d - 2^{d-k}.
\end{equation}
\end{itemize}

Now examine a subcube $G \cap Q_1(x_i = 0)$ for an $i \in [2;k+1]$. Since $G \cap Q_0(x_i = 0) = \varnothing$, we have $\delta(G \cap Q_1(x_i = 0)) \geq d-1$. By Lemma \ref{thm:minsize}, either $|V(G \cap Q_1(x_i = 0))| = 2^{d-1}$ or $|V(G \cap Q_1(x_i = 0))| \geq \min_{t = d}^{n-1} S(t,d-1)$.

\begin{itemize}
\item Assume it is the latter case; by the induction of this section, we have that the minimum is achieved by $t = d$ with $|V(G \cap Q_1(x_i = 0))| \geq 2^d - 2^{2(d-1)-d} = 2^d - 2^{d-2}$.

We can now assume that $G \cap Q_1(x_i = 1, x_j = 0) \neq \varnothing$, for $j \in [2;k+1]$, $i \neq j$. Since $G \cap Q_0(x_i = 1, x_j = 0) = \varnothing$, we have $\delta(G \cap Q_1(x_i = 1, x_j = 0)) \geq d-2$ and by Theorem \ref{thm:simon} we have $|V(G \cap Q_1(x_i = 1, x_j = 0))| \geq 2^{d-2}$. Thus
\begin{align}
R &\geq |V(G \cap Q_1(x_i = 0))| + |V(G \cap Q_1(x_i = 1, x_j = 0))| \geq \\
&\geq (2^d - 2^{d-2}) + 2^{d-2} = 2^d > 2^d - 2^{d-k}.
\end{align}

\item Otherwise it is the former case in Lemma \ref{thm:minsize} for each $i$, $|V(G \cap Q_1(x_i = 0))| = 2^{d-1}$. By Theorem \ref{thm:simon}, $G \cap Q_1(x_i = 0)$ must be a $(d-1)$-subcube of $Q_n$.

Pick $i_1, i_2 \in [2;k+1]$, $i_1 \neq i_2$. We have that $Q_a = G \cap Q_1(x_{i_1}= 0)$ and $Q_b = G \cap Q_1(x_{i_2}= 0)$ are both $(d-1)$-subcubes. We can now assume that $G \cap Q_1(x_{i_1} = 0, x_{i_2} = b) \neq \varnothing$, for $b \in \{0,1\}$. This means that the $i_2$-th bit is not fixed for the subcube $Q_a$. Thus $G \cap Q_1(x_{i_1}= 0,x_{i_2}=0)$ is a $(d-2)$-subcube. Hence $Q_a$ and $Q_b$ overlap exactly in a $(d-2)$-subcube.

Examine $Q_a \cap Q_b$. Two of its fixed bits are the $i_1$-th and the $i_2$-th, which are distinct positions. Thus it has $n-(d-2)-2=n-d$ fixed positions not in $[2;k+1]$. Let the $d$-subcube defined by these restrictions be $C$. As $Q_a$ and $Q_b$ are both $(d-1)$-subcubes, they must share these $n-d$ fixed positions. As this applies for any $i_1 \neq i_2$, we have that $G \cap Q_1(x_i = 0) \subset C$ for any $i \in [2;k+1]$.

We show that $C \subset G$. Pick $x \in C$. Suppose for some $i \in [2;k+1]$, we have $x_i = 0$. Then $x \in G \cap Q_1(x_i = 0)$. Otherwise we have $x_i = 1$ for each $i \in [2;k+1]$. But then $x \in D$.\footnote{In this case, we have obtained that $G$ is a union of two $d$-dimensional subcubes $D$ and $C$, such that each bit position is fixed in at most one of them. This is essentially the same construction as given in subsection \ref{sec:sufficiency}.}

Examine the intersection of $D$ and $C$. Each position of $[2;k+1]$ is fixed in $D$ ($k$ positions). On the other hand, $n-d$ more positions not in $[2;k+1]$ are fixed in $C$. Thus their intersection is a $(d-k)$-subcube, and $R = 2^d - 2^{d-k}$.\hfil\qedhere
\end{itemize}
\end{proof}

Since $k = n-d$, we have $R \geq 2^d-2^{2d-n}$. Ultimately we get
\begin{align}
|V| &= |V(D)| + R \geq 2^d + (2^d - 2^{2d-n}) = 2^{d+1} - 2^{2d-n}.
\end{align}
This completes the proof of Theorem \ref{thm:main}.
\qed
\end{enumerate}

\subsection{Application for Boolean Functions}

Theorem \ref{thm:main} implies:

\begin{repcorollary}{thm:result}
Let $f(x)$ be a Boolean function on $n$ variables. If $\forall i \in [n] \, \forall b \in \{0, 1\} \, \exists x \, (x_i = b, f(x) = 1)$, then
\begin{equation}
| \{ x \mid f(x) = 1 \} | \geq 2^{n-s_1(f)+1} - 2^{n-2s_1(f)}.
\end{equation}
\end{repcorollary}

\begin{proof}
Let $G$ be a subgraph of $Q_n$ induced by the set of vertices $V = \{ x \mid f(x) = 1 \}$. The minimum degree of $G$ is $\delta(G) = n - s_1(f)$. The given constraint means that $G$ is irreducible. Then, by Theorem \ref{thm:main}, 
\begin{equation}
|V| \geq 2^{(n-s_1(f))+1} - 2^{2(n-s_1(f))-n} = 2^{n-s_1(f)+1} - 2^{n-2s_1(f)}.
\end{equation}
\end{proof}

\subsection{Generalization of Simon's Lemma}

We use Theorem \ref{thm:main} and Lemma \ref{thm:minsize} to prove Theorem \ref{thm:stronger}, which is a stronger version of Simon's lemma (Theorem \ref{thm:simon}):

\begin{reptheorem}{thm:stronger}
Let $G = (V, E)$ be a non-empty subgraph of $Q_n$ ($n \geq 0$) with $\delta(G) = d$. Then either $|V| = 2^d$ or $|V| \geq \frac 3 2 \cdot 2^d$, with $V=|2^d|$ achieved if and only if $G$ is a $d$-subcube.
\end{reptheorem}

\begin{proof}
By Theorem \ref{thm:main} we may substitute $\left\lceil 2^{d+1} - 2^{2d-n}\right\rceil$ instead of $S(n,d)$ in Lemma \ref{thm:minsize}. Then in 
\begin{equation}
\min_{i=d+1}^n S(i,d) = \min_{i=d+1}^n \left\lceil 2^{d+1} - 2^{2d-i}\right\rceil
\end{equation}
the minimum is obtained for $i = d+1$. Thus either $|V| = 2^d$ or $|V| \geq 3 \cdot 2^{d-1}$.
\end{proof}

\section{Conclusion}

In this paper, we have shown two results on the structure of low sensitivity subsets of Boolean hypercube:
\begin{itemize}
\item
Theorem \ref{thm:main}: a tight lower bound on the size of irreducible low sensitivity sets 
$S\subseteq \{0, 1\}^n$, that is, sets $S$ that are not contained in any subcube of $\{0, 1\}^n$ obtained by fixing one or more variables $x_i$;
\item
Theorem \ref{thm:stronger}: a gap theorem that shows that $S\subseteq\{0, 1\}^n$ of sensitivity $s$ must either have $|S|=2^{n-s}$ or $|S|\geq \frac{3}{2} 2^{n-s}$.
\end{itemize}
The gap theorem follows from the first result by classifying $S\subseteq\{0, 1\}^n$ into irreducible sets and sets that are constructed from irreducible 
subsets $S'\subseteq \{0, 1\}^{n-k}$ for some $k\in\{1, 2, \ldots, s\}$ and then using the first result for each of those categories. We find this gap theorem quite surprising.

Both results contribute to understanding the structure of low-sensitivity subsets of the Boolean hypercube. 
After this paper was completed, 
we have used the gap theorem to obtain a new upper bound on block sensitivity in terms of sensitivity:
\begin{equation}
bs(f) \leq \max\left(2^{s(f)-1}\left(s(f)-\frac 1 3\right), s(f) \right).
\end{equation}
We report this result in \cite{Ambainis_Prusis_Vihrovs}.

\bibliographystyle{abbrv}
\bibliography{bibliography}

\end{document}